\documentclass[letterpaper, 10 pt, conference]{ieeeconf}
\usepackage{graphicx}
\usepackage{amsmath}
\usepackage{subfigure}
\usepackage{psfrag}
\usepackage{epsfig}
\usepackage{amssymb}
\usepackage{latexsym}
\usepackage{amsmath}
\input{epsf}
\include{epsf}

\newtheorem{lemma}{Lemma}
\newtheorem{proposition}{Proposition}
\newtheorem{algorithm}{Algorithm}
\newtheorem{assumption}{Assumption}
\newtheorem{corollary}{Corollary}

\newcommand{\beq}{\begin{equation}}
\newcommand{\eeq}{\end{equation}}
\newcommand{\beqr}{\begin{equation}\begin{array}{l}}
\newcommand{\eeqr}{\end{array}\end{equation}}
\newcommand{\beqa}{\begin{eqnarray}}

\newcommand{\eeqa}{\end{eqnarray}}

\def\qed{\hfill\rule[-1pt]{5pt}{5pt}\par\medskip}

\newcommand{\bea}{\begin{eqnarray}}
\newcommand{\eea}{\end{eqnarray}}

\def\qed{\hfill\rule[-1pt]{5pt}{5pt}\par\medskip}

\title{\LARGE \bf
Algorithm for Optimal Mode Scheduling in  Switched Systems}
\author{Y. Wardi and M. Egerstedt
\thanks{ School of
Electrical and Computer Engineering, Georgia Institute of Technology,
Atlanta, GA 30332, USA. Email: \{ywardi,magnus\}@ece.gatech.edu}
}

\begin{document}
\maketitle
\thispagestyle{empty}
\pagestyle{empty}
\begin{abstract}
This paper considers the problem of computing the schedule of modes in a switched dynamical system, that minimizes
a cost functional defined on the trajectory of the system's continuous state variable. A recent approach  to such optimal control problems consists of 
algorithms that alternate  between computing the optimal switching times between modes in a given sequence, and updating the mode-sequence by inserting to it
a finite number of new modes.
 These algorithms have an inherent inefficiency due to their sparse update of the mode-sequences, while
spending most of the computing times on optimizing with respect to the switching times for a given mode-sequence. This paper proposes an
algorithm that operates directly in the schedule space without resorting to the timing optimization problem.
It is based on the Armijo step size along certain G\^ateaux derivatives of the performance functional, thereby avoiding some of the computational
difficulties associated with discrete scheduling parameters. Its convergence to local minima as well as its rate of convergence are proved,
and a simulation example on a nonlinear system exhibits quite a fast convergence.
\end{abstract}

\section{Introduction}
Switched-mode hybrid dynamical systems often are characterized by the following equation,
\begin{equation}
\dot{x}\ =\ f(x,v),
\end{equation}
where $x\in R^n$ is the state variable, $v\in V$ with $V$ being a given finite set,
and $f:R^n\times V\rightarrow R^n$ is a suitable function. Suppose that the system evolves on a horizon-interval
$[0,T]$ for some $T>0$, and  that the initial state $x(0)=x_{0}$ is given for some $x_{0}\in R^n$. The input control of this system, $v(t)$, is
discrete since $V$ is a finite set, and  we assume that
the function $v(t)$ changes its values a finite number of times during the horizon interval
$[0,T]$.

Such systems have been investigated in the past several years due to their relevance in
control applications such as  mobile robotics \cite{Egerstedt00}, vehicle control \cite{Wang97}, switching circuits \cite{Almer10} and
references therein,
telecommunications \cite{Rehbinder00,Hristu01}, and situations where a controller has to switch its attention among multiple subsystems
 \cite{Lincoln01} or data sources \cite{Brockett95}.
Of a particular interest in these applications is an optimal control problem where it is desirable to
minimize a cost functional (criterion)
of the form
\begin{equation}
J\ :=\ \int_{0}^{T}L(x)dt
\end{equation}
for a given $T>0$,
where $L:R^n\rightarrow R$ is a cost function defined on the state  trajectory.

This general nonlinear optimal-control problem was formulated in \cite{Branicky98}, where the particular values of $v\in V$ are associated with
the various modes of the system.\footnote{The setting in \cite{Branicky98} is more general since it involves a continuous-time control
$u\in R^k$ as well as a discrete control $v$. In this paper we focus only on the discrete control since it captures the salient points of switched-mode systems, and we
defer discussion of the general case to a forthcoming publication.}
 Several variants of the
maximum principle were derived for this problem in \cite{Sussmann99,Piccoli98,Shaikh02}, and subsequently  provably-convergent optimization algorithms  were developed in
\cite{Xu02,Shaikh02,Shaikh07,Egerstedt06,Attia05}. We point out that two kinds of problems were considered: those where the sequence of modes is fixed and the
controlled variable consists of the switching times between
them, and those where the controlled variable is comprised of the sequence of modes as well as the switching times between them.
We call the former problem the {\it timing optimization problem}, and the latter problem,
the {\it scheduling optimization problems}.

The timing optimization problem generally is simpler than the scheduling optimization problem since essentially it  is a nonlinear-programming
problem (albeit with a special structure) having only  continuous variables, while the scheduling problem has a discrete sequencing-variable as well. Furthermore, scheduling problems generally are
NP hard, and computational techniques have to search for solutions that are suboptimal in a suitable sense.  Thus, while the algorithms that were proposed early
focused on the timing optimization problem, several different (and apparently complementary) approaches to
the scheduling-optimization problem have emerged as well.  Zoning algorithms that compute (iteratively) the mode
sequences based on geometric properties of the problem have been developed in \cite{Shaikh05}, needle-variations techniques were presented in
  \cite{Axelsson08}, and relaxation methods were proposed in  \cite{Caldwell10}.  In contrast, the
  algorithm considered in this paper computes its iterations
  directly in the schedule space without resorting to relaxations, and as argued later in the sequel, may compute optimal (or suboptimal) schedules quite
  effectively.

Our stating point is the algorithm  we developed in \cite{Axelsson08} which alternates between the following two steps: (1). Given a sequence of modes,
compute
the switching times among them that minimize the functional $J$. (2). Update the mode-sequence by inserting to it a single mode at a (computed) time
that would lead to the greatest-possible reduction rate in $J$. Then repeat Step 1, etc.

The second step deserves some explanation. Fix a  time $t\in[0,T]$, and let us denote the system's mode at that time by $M_{\alpha}$. Now
suppose that we replace this mode by another mode, denoted by $M_{\beta}$, over the time-interval $[t,t+\lambda]$ for some given $\lambda>0$,
and denote by $\tilde{J}(\lambda)$ the cost functional  $J$ defined by (2) as a function of
$\lambda$. We call the one-sided derivative $\frac{d\tilde{J}}{d\lambda^+}(0)$  the {\it insertion gradient}, and we note that if
$\frac{d\tilde{J}}{d\lambda^+}(0)<0$ then inserting $M_{\beta}$ for a brief amount of time at time $t$ would result in a decrease in $J$, while
if $\frac{d\tilde{J}}{d\lambda^+}(0)>0$ then such an insertion would result in an increase in $J$. Now the second step of the
 algorithm computes the time $t\in [0,T]$ and mode $M_{\beta}$ that minimize  the insertion gradient, and it performs the insertion accordingly. We mention that if
 the insertion gradient is non-negative for every mode $M_{\beta}$ and time $t\in[0,T]$ then the schedule in question satisfies
 a necessary optimality condition and no insertion is performed.

 The aforementioned algorithm has a peculiar feature in that it solves a timing optimization problem between consecutive mode-insertions.
 This feature appears awkward and suggests that the algorithm can be quite inefficient, but it is required for the convergence-proof
 derived in \cite{Axelsson08}. In fact, that proof breaks down if the insertions are made for schedules that do not necessarily
 comprise solution points of the timing optimization problem for their given mode-sequences. The reason seems to be in the fact that the
 insertion gradient
 is not continuous in the time-points at which the insertion of a given mode $M_{\beta}$ is made.
 However, this lack of continuity can be overcome by other properties of the problem at hand, and this leads to the development
 of the algorithm  that is proposed in this paper, which appears to be more efficient than the one in \cite{Axelsson08}.

 The algorithm we describe here computes its iterations directly in the space of mode-schedules without having to solve any timing optimization
 problems. Furthermore, at each iteration it switches the mode not at a finite set of times,   but at sets comprised of  unions of positive-length
intervals in the time-horizon $[0,T]$. The algorithm  is based on the idea of the Armijo step size used in gradient-descent techniques \cite{Polak97}, and it uses
the Lebesgue measure of sets where the modes are to be changed as the step-size parameter. To the best of our knowledge this idea has not been used in
extant algorithms for optimal control problems, and while it appears natural in the setting of switched-node systems, it may have extensions to
other optimal-control settings as well. We prove the algorithm's convergence and its convergence-rate, which we show to be independent of the number of intervals where the modes are changed at a given iteration.

The rest of the paper is organized as follows. Section II sets the mathematical formulation of the problem and recounts some established results. Section III
carries out the analysis, while Section IV presents a simulation example. Finally, Section V concludes the paper.

\section{Problem Formulation and Survey of Relevant Results}
Consider the state equation (1) and recall that the initial state $x_{0}$ and the final time $T>0$ are given.
We make the following assumption regarding the vector field $f(x,v)$ and the state trajectory $\{x(t)\}$.
\begin{assumption}
(i). For every $v\in V$, the function $f(x,v)$ is twice-continuously differentiable ($C^2$) throughout $R^n$. (ii).
The state trajectory $x(t)$ is continuous at all $t\in[0,T]$.
\end{assumption}

Every mode-schedule is associated with an input control function $v:[0,T]\rightarrow V$, and we define an admissible mode schedule to
be a schedule whose associated control function $v(\cdot)$ changes its values a finite number of times
throughout the interval $t\in[0,T]$. We denote the space of admissible schedules by $\Sigma$, and a typical admissible schedule by
$\sigma\in\Sigma$.
Given  $\sigma\in\Sigma$, we define the length of $\sigma$ as
the number of  consecutive different values of $v$ on the horizon interval $[0,T]$, and denote it  by $\ell(\sigma)$. Furthermore, we denote the $ith$ successive value of $v$ in
$\sigma$ by
$v^i$, $i=1,\ldots,\ell(\sigma)$, and the switching time between $v^i$ and $v^{i+1}$ will be denoted by  $\tau_{i}$.
Further defining $\tau_{0}:=0$ and $\tau_{\ell(\sigma)}=T$, we observe that the input control  function is defined by
 $v(t)=v_{i}$ $\forall i\in[\tau_{i-1},\tau_{i})$,
$i=1,\ldots,\ell(\sigma)$. We require that $\ell(\sigma)<\infty$ but impose no upper bound on $\ell(\sigma)$.

Given $\sigma\in\Sigma$, define the costate $p\in R^n$ by the following differential equation,
\begin{equation}
\dot{p}\ =\ -\Big(\frac{\partial f}{\partial x}(x,v)\Big)^Tp-\Big(\frac{dL}{dx}(x)\Big)^T
\end{equation}
with the boundary condition $p(T)=0$.  Fix time $s\in[0,T)$,  $w\in V$, and $\lambda>0$, and consider replacing the value of $v(t)$
by $w$ for every $t\in[s,s+\lambda)$. This amounts to changing the mode-sequence $\sigma$ by
inserting the mode associated with $w$ throughout the interval $[s,s+\lambda)$. Denoting by $\tilde{J}(\lambda)$ the value of the cost
functional resulting from this insertion, the insertion gradient is defined by $\frac{d\tilde{J}}{d\lambda^+}(0)$.
Of course this insertion gradient depends on the mode-schedule $\sigma$, the inserted mode associated with $w\in V$, and the insertion time
$s$, and hence we denote it by
$D_{\sigma,s,w}$. We have the following result (e.g., \cite{Egerstedt06}):
\begin{equation}
D_{\sigma,s,w}\ =\ p(s)^T\big(f(x(s),w)-f(x(s),v(s))\big).
\end{equation}
As mentioned earlier, if $D_{\sigma,s,w}<0$ then inserting to $\sigma$
the mode associated with $w$ on a small interval  starting at time  $s$ would reduce the cost functional. On the other hand, if
$D_{\sigma,s,w}\geq 0$ for all $w\in V$ and $s\in[0,T]$
then we can think of $\sigma$ as satisfying a local optimality condition.
Formally, define $D_{\sigma,s}:=\min\{D_{\sigma,s,w}:w\in V\}$, and define $D_{\sigma}:=\inf\{D_{\sigma,s}:s\in[0,T]\}$.
Observe that $D_{\sigma,s,v(s)}=0$ since $v(s)$ is associated with the same mode at time $s$ and hence $\sigma$ is not modified,
and consequently, by definition, $D_{\sigma,s}\leq 0$ and $D_{\sigma}\leq 0$ as well. The condition $D_{\sigma}=0$
is a natural first-order necessary optimality condition, and the purpose of the algorithm described below is to compute a
mode-schedule $\sigma$ that satisfies it.

Our algorithm is a descent method based on the principle of the Armijo step size. Given a schedule $\sigma\in\Sigma$, it computes the next
schedule, $\sigma_{next}$, by changing the modes associated with points $s\in[0,T]$ where $D_{\sigma,s}<0$. The main point of departure from
existing algorithms (and especially those in \cite{Axelsson08}) is that the set of  such points $s$ is not finite or discrete, but has a positive
Lebesgue measure. Moreover, the Lebesgue measure of this set acts as the parameter for the Armijo procedure.

Now one of the basic requirements of algorithms in the general setting of nonlinear programming is that
 every accumulation
 point of a computed sequence of iteration points   satisfies a certain optimality condition, like stationarity
 or the Kuhn-Tucker condition. However, in our case such a convergence property is meaningless since the
 schedule-space $\Sigma$ is neither finite dimensional nor complete, the latter due to the requirement that $\ell(\sigma)<\infty$
 $\forall\ \sigma\in\Sigma$. Consequently convergence of our algorithm has to be characterized by other means, and to this
 end we use Polak's concept of minimizing sequences \cite{Polak84}. Accordingly,
 the quantity $D_{\sigma}$ acts as an {\it optimality function} \cite{Polak97}, namely the optimality condition in question is $D_{\sigma}=0$, while
 $|D_{\sigma}|$ indicates an extent to which $\sigma$ fails to satisfy that optimality condition.
 Convergence of an algorithm means that, if it computes a sequence of schedules $\{\sigma_{k}\}_{k=1}^{\infty}$ then,
\begin{equation}
 \limsup_{k\rightarrow\infty}D_{\sigma_{k}}\ =\ 0;
\end{equation}
in some cases the stronger condition $\lim_{k\rightarrow\infty}D_{\sigma_{k}}\ =\ 0$ applies. In either case, for every $\epsilon>0$
 the algorithm yields
an admissible mode-schedule $\sigma\in\Sigma$ satisfying  the inequality $D_{\sigma}>-\epsilon$.
 Our analysis will yield Equation (5)  by proving a uniformly-linear convergence rate of the algorithm.\footnote{The  reason for the ``limsup''
 in (5) instead of the stronger form of convergence (with ``lim'' instead of ``limsup'') is due to technical peculiarities of the
 optimality function $D_{\sigma}$ that will be discussed later. We will argue that the stronger form of convergence applies except for pathological
  situations. Furthermore, we will define an alternative optimality function and prove the stronger form of convergence for it.
  The choice of the most-suitable optimality function is largely theoretical and will not be addressed in this paper.}

 Since the Armijo step-size technique will play a key role in our  algorithm, we conclude this section with
 a recount of its main features. Consider the general setting of nonlinear programming where it is desirable
 to minimize  a $C^{2}$ function $f:R^n\rightarrow R$, and suppose that the
 Hessian
 $\frac{d^2f}{dx^2}(x)$ is bounded on $R^n$. Given $x\in R^n$, a steepest descent from $x$ is any vector in the direction $-\nabla f(x)$;
 we normalize the  gradient by defining $h(x):=\frac{\nabla f(x)}{||\nabla h(x)||}$, and call $-h(x)$  the steepest-descent  direction.
 Let $\lambda(x)\geq 0$ denote the step size so that the next point computed
 by the algorithm, denoted by $x_{next}$, is defined as
\begin{equation}
x_{next}\ =\ x-\lambda(x)h(x).
\end{equation}
 The Armijo step size procedure defines $\lambda(x)$ by an approximate line minimization in the following way
 (see \cite{Polak97}):
Given constants $\alpha\in(0,1)$ and $\beta\in(0,1)$, define the integer
$j(x)$ by
\begin{eqnarray}
j(x)\ :\ \min\Big\{j=0,1,\ldots,\ :\nonumber \\
 f(x-\beta^j\nabla f(x))-f(x)\leq-\alpha\beta^j||\nabla f(x)||^2\Big\},
\end{eqnarray}
and define
\begin{equation}
\lambda(x)\ =\ \beta^{j(x)}||\nabla f(x)||.
\end{equation}
Now the  {\it steepest descent algorithm with Armijo step size} computes a sequence of iteration points
$x_{k}$, $k=1,2,\ldots,$ by the formula $x_{k+1}=x_{k}-\lambda(x_{k})h(x_{k})$; $\lambda(x_{k})$ is called the Armijo step size at $x_{k}$.
The main convergence property of this algorithm \cite{Polak97} is
that every accumulation point $\hat{x}$ of a computed sequence $\{x_{k}\}_{k=1}^{\infty}$ satisfies the stationarity condition $\nabla f(\hat{x})=0$.
Several results concerning convergence rate have been derived as well, and the  one of interest to us
is given by Proposition 1 below. Its proof  is  contained in the arguments of the proof of Theorem 1.3.7 and especially Equation (8b) in \cite{Polak97},
but since we have not seen the result  stated in the same way as in Proposition 1, we provide a brief proof in the appendix.
\begin{proposition}
Suppose that $f(x)$ is $C^{2}$, and that there exists a constant $L>0$ such that, for
every $x\in R^n$, $||H(x)||\leq L$, where $H(x):=\frac{df^2}{dx^2}(x)$. Then the following two statements are true:
(1). For every $x\in R^n$ and for every
$\lambda\geq 0$ such that $\lambda\leq\frac{2}{L}(1-\alpha)||\nabla f(x)||$,
\begin{equation}
f(x-\lambda h(x))-f(x)\ \leq\ -\alpha\lambda||\nabla f(x)||.
\end{equation}
(2). For every $x\in R^n$,
\begin{equation}
\lambda(x)\ \geq\ \frac{2}{L}\beta(1-\alpha)||\nabla f(x)||.
\end{equation}
\qed
\end{proposition}
This implies the following convergence result:
\begin{corollary}
(1). There exists $c>0$ such that $\forall x\in R^n$,
\begin{equation}
f(x_{next})-f(x)\ \leq\ -c||\nabla f(x)||^2.
\end{equation}
(2). If the algorithm computes a bounded sequence $\{x_{k}\}_{k=1}^{\infty}$ then
\begin{equation}
\lim_{k\rightarrow\infty}\nabla f(x_{k})\ =\ 0.
\end{equation}
\begin{proof}
(1). Define  $c:=\frac{2}{L}\alpha(1-\alpha)\beta$. Then (11) follows directly from Equations (9) and (10).

(2). Follows immediately from part (1) and the fact that the sequence $\{f(x_{k})\}_{k=1}^{\infty}$ is monotone non-increasing.
\end{proof}
\end{corollary}

\section{Algorithm for   Mode-Scheduling Minimization}
To simplify the notation and analysis we assume first that the set $V$ consists only of two elements, namely the
system is bi-modal. This assumption incurs no significant loss of generality, and at the end of this section we will point out
an extension to the general case where   $V$ consists of an arbitrary finite number of points.
Let us denote the two elements of $V$ by $v_{1}$ and $v_{2}$. A mode-schedule $\sigma$  alternates between these two points,
and we denote by $\{v^1,\ldots,v^{\ell(\sigma)}\}$ the sequence of values of $v$ associated with the mode-sequence comprising  $\sigma$.
Denoting by $v^c$ the complement of $v$, we have that  $v^{i+1}=(v^i)^c$ for all $i=1,\ldots,\ell(\theta)-1$.

Consider a mode-schedule $\sigma\in\Sigma$ that does not satisfy the necessary optimality condition,
 namely $D_{\sigma}<0$. Define the set $S_{\sigma,0}$ as
 $S_{\sigma,0}:=\{s\in[0,T]\ :\ D_{\sigma,s}<0\}$, and note that  $S_{\sigma,0}\neq\emptyset$. Recall that $v(s)$ denotes the value
 of $v$ at the time $s$. Then for every  $s\in S_{\sigma,0}$ which is not a switching time,
 an insertion  of the  complementary mode  $v(s)^c$  at $s$ for a small-enough period would result in a decrease of $J$.
 Our goal is to flip the modes (namely, to switch them
 to their complementary ones)  in a large subset of $S_{\sigma,0}$ that would result
 in a substantial decrease in $J$, where by the term  ``substantial decrease'' we mean a decrease by at least
 $aD_{\sigma}^2$ for some constant $a>0$. This ``sufficient descent'' in $J$ is akin to the descent property of the
 Armijo step size as reflected in Equation (11).

 This sufficient-descent property cannot be guaranteed by flipping the mode at every time $s\in S_{\sigma,0}$.  Instead, we search for
 a subset of $S_{\sigma,0}$  where, flipping the mode at every $s$ in that subset would guarantee
 a sufficient descent. This subset will consist os points $s$ where $D_{\sigma,s}$ is ``more negative'' than at typical
 points $s\in S_{\sigma,0}$.
 Fix $\eta\in(0,1)$ and define
 the set $S_{\sigma,\eta}$ by
 \begin{equation}
 S_{\sigma,\eta}\ =\ \big\{s\in[0,T]\ :\ D_{\sigma,s}\leq\eta D_{\sigma}\big\}.
 \end{equation}
 Obviously $S_{\sigma,\eta}\neq \emptyset$ since  $D_{\sigma}<0$.
Let $\mu(S_{\sigma,\eta})$ denote the Lebesgue measure of $S_{\sigma,\eta}$, and  more generally, let
$\mu(\cdot)$ denote the Lebesgue measure on $R$. For every  subset $S\subset S_{\sigma,\eta}$, consider
flipping the mode at every point $s\in S$, and denote by $\sigma(S)$ the resulting mode-schedule.  In the forthcoming we will search for
a set $S\subset S_{\sigma,\eta}$ that will give us the desired sufficient descent.

Fix $\eta\in(0,1)$. Let  $S:[0,\mu(S_{\sigma,\eta})]\rightarrow  2^{S_{\sigma,\eta}}$ (the latter object is the set  of
subsets of $S_{\sigma,\eta}$) be a mapping
 having the following two properties:
(i) $\forall\lambda\in [0,\mu(S_{\sigma,\eta})]$, $S(\lambda)$ is the finite union of closed intervals; and
(ii) $\forall\lambda\in[0,\mu(S_{\sigma,\eta})]$, $\mu(S(\lambda))=\lambda$.
We define $\sigma(\lambda)$ to be the mode-schedule obtained from
$\sigma$ by flipping the mode at every time-point
$s\in S(\lambda)$.
 For example, $\forall\lambda\in[0,\mu(S_{\sigma,\eta})]$ define
$s(\lambda):=\inf\{s\in S_{\sigma,\eta}:\mu([0,s]\cap S_{\sigma,\eta})=\lambda\}$, and define
$S(\lambda):=[0,s(\lambda)]\cap S_{\sigma,\eta}$. Then $\sigma(\lambda)$ is the schedule obtained
from $\sigma$ by flipping the modes lying in the leftmost subset of
$S_{\sigma,\eta}$ having  Lebesgue-measure $\lambda$, and it is the finite
union of closed intervals if so is $S_{\sigma,\eta}$.

We next  use such a mapping $S(\lambda)$ to define an Armijo step-size procedure for computing
a schedule $\sigma_{next}$ from $\sigma$.
Given constants $\alpha\in(0,1)$ and  $\beta\in(0,1)$, in addition to $\eta\in(0,1)$. Consider a given $\sigma\in\Sigma$ such that
$D_{\sigma}<0$.
For every
$j=0,1,\ldots$, define $\lambda_{j}:=\beta^{j}\mu(S_{\sigma,\eta})$, and
define $j(\sigma)$ by
\begin{equation}
j(\sigma):=\min\Big\{j=0,1,\ldots,\ :\ J(\sigma(\lambda_{j}))-J(\sigma)\leq\alpha\lambda_{j}D_{\sigma}\Big\}.
\end{equation}
Finally, define $\lambda(\sigma):=\lambda_{j(\sigma)}$, and set
$\sigma_{next}:=\sigma(\lambda(\sigma))$.

Observe that the Armijo step-size procedure is applied here not to the steepest descent (which is not defined in our problem setting) but to
a descent direction defined by a G\^ateaux derivative of $J$ with respect to a subset of the interval $[0,T]$ where the modes are to be flipped.
Generally this G\^ateux derivative is not necessarily continuous in $\lambda$ and hence the standard arguments for sufficient descent do not apply.
However, the problem has a special structure guaranteeing sufficient descent and the algorithm's convergence in
the sense of minimizing sequences. Furthermore, the sufficient descent property depends on $\mu(S_{\sigma,\eta})$ but is independent of both the string size
$\ell(\sigma)$ and the particular choice of
the mapping $S:[0,\mu(S_{\sigma,\eta})]\rightarrow 2^{S_{\sigma,\eta}}$. This guarantees that the convergence rate of the algorithm is not
reduced when the string lengths of the schedules computed in successive iterations grows
unboundedly.

We next present the algorithm formally.  Given constants $\alpha\in(0,1)$, $\beta\in(0,1)$,
and $\eta\in(0,1)$. Suppose that for every $\sigma\in\Sigma$ such that $D_{\sigma}<0$ there exists a mapping
$S:[0,\mu(S_{\sigma,\eta})]\rightarrow 2^{S_{\sigma,\eta}}$ with the  aforementioned properties.
\begin{algorithm}
{\it Step 0:} Start with an arbitrary schedule $\sigma_{0}\in\Sigma$. Set $k=0$.\\
{\it Step 1:} Compute $D_{\sigma_{k}}$. If $D_{\sigma_{k}}=0$, stop and exit; otherwise, continue.\\
{\it Step 2:} Compute $S_{\sigma_{k},\eta}$ as defined in (13), namely $S_{\sigma_{k},\eta}\ =\ \{s\in[0,T]\ :\ D_{\sigma_{k},s}\leq\eta D_{\sigma_{k}}\}$. \\
{\it Step 3:} Compute $j(\sigma_{k})$ as defined by (14), namely
\begin{eqnarray}
j(\sigma_{k})\ =\nonumber \\
 \min\Big\{j=0,1,\ldots,\ :\ J(\sigma_{k}(\lambda_{j}))-J(\sigma_{k})\leq \alpha\lambda_{j}D_{\sigma_{k}}\Big\}
\end{eqnarray}
with $\lambda_{j}:=\beta^j\mu(S_{\sigma_{k},\eta})$,
and set $\lambda(\sigma_{k}):=\lambda_{j(\sigma_{k})}$.\\
{\it Step 4:} Define  $\sigma_{k+1}:=\sigma_{k}(\lambda(\sigma_{k}))$,  namely the schedule obtained
from $\sigma_{k}$ by flipping the mode at every time-point $s\in S(\lambda(\sigma_{k}))$. Set
$k=k+1$, and go to Step 1.
\end{algorithm}

It must be mentioned that the computation of the set $S_{\sigma_{k},\eta}$ at Step 2 typically
requires an adequate approximation. This paper  analyzes the algorithm under the
assumption of an exact computation of $S_{\sigma_{k},\eta}$, while the case involving adaptive precision will
be treated in a later, more comprehensive publication.

The forthcoming analysis is carried out under Assumption 1, above.
It requires the following two preliminary results, whose proofs follow as  corollaries from established results on sensitivity analysis
of
solutions to differential equations \cite{Polak97}, and hence are relegated to the appendix.
 Given $\sigma\in\Sigma$,
consider an interval $I:=[s_{1},s_{2}]\subset [0,T]$ of a positive length, such that
the modes associated with all $s\in I$ are the same, i.e.,
$v(s)=v(s_{1})$ $\forall s\in I$. Denote by $\sigma_{s_{1}}(\gamma)$ the mode-sequence
obtained from $\sigma$ by flipping the modes at every time $s\in[s_{1},s_{1}+\gamma]$,
and consider the resulting cost function $J(\sigma_{s_{1}}(\gamma))$ as a function of
$\gamma\in [0,s_{2}-s_{1}]$.
\begin{lemma}
There exists a constant $K>0$ such that,
for every $\sigma\in\Sigma$, and for every interval $I=[s_{1},s_{2}]$ as above, the function
$J(\sigma_{s_{1}}(\cdot))$ is twice-continuously differentiable $(C^2)$ on the interval
$\gamma\in[0,s_{2}-s_{1}]$; and for every $\gamma\in[0,s_{2}-s_{1}]$,
$|J(\sigma_{s_{1}}(\gamma))^{''}|\leq K$ (``prime'' indicates derivative with respect to $\gamma)$.
\end{lemma}
\begin{proof}
Please see the appendix.
\end{proof}

We remark that the $C^2$ property of $J(\sigma_{s_{1}}(\cdot))$ is in force only as long as $v(s)=v(s_{1})$ $\forall s\in[s_{1},s_{2}]$. The second
assertion of the above lemma does not quite follow from the first one; the bound $K$ is independent of the specific interval
$[s_{1},s_{2}]$.

Lemma 1  in conjunction with Corollary 1 (above)    can yield  sufficient descent only in a local sense, as
long as the same mode is scheduled according to $\sigma$. At mode-switching times
$D_{\sigma,s}$ is no longer continuous in $s$, and hence Lemma 1 cannot be extended
to intervals where $v(\cdot)$ does not have a constant value. Nonetheless we can prove the
sufficient-descent property in a more global sense with the aid of the following result, whose validity  is
due to the
 special structure of the problem.
\begin{lemma}
There exists a constant $K>0$ such that for every $\sigma\in\Sigma$, for every interval $I=[s_{1},s_{2}]$ as above (i.e.,
such that $\sigma$ has the same mode throughout $I$),
for
every $\gamma\in[0,s_{2}-s_{1})$, and for every $s\geq s_{2}$,
\begin{equation}
|D_{\sigma_{s_{1}}(\gamma),s}-D_{\sigma,s}|\ \leq\ K\gamma.
\end{equation}
\end{lemma}
\begin{proof}
Please see the appendix.
\end{proof}

To explain this result, recall that $\sigma_{s_{1}}(\gamma)$ is the mode-schedule obtained from $\sigma$ by flipping all the modes
on the interval $[s_{1},s_{1}+\gamma]$. Thus, Equation (16) provides an upper bound on the magnitude of the difference between  the insertion gradients
of the sequences $\sigma$ and $\sigma_{s_{1}(\gamma)}$ at the same point $s$. Furthermore, Lemma 2 implies
 a uniform Lipschitz continuity of the insertion gradient
at every point $s>s_{2}$  with respect to the length of the insertion interval
  $\gamma$.  This is not the same as
continuity of $D_{\sigma,s}$ with respect to $s$, which we know is not true.

Recall the following terminology: given $\sigma\in\Sigma$ and $S\subset[0,T]$,  $\sigma(S)$ denotes the schedule obtained
by flipping the mode of $\sigma$ at every $\tau\in S$.
\begin{corollary}
There exists $K>0$ such that, for every $\sigma\in\Sigma$,
for every subset $S\subset[0,T]$ comprised of a finite number of intervals, and for every
$s\geq\sup\{\tilde{s}\in S\}$,
\begin{equation}
|D_{\sigma(S),s}-D_{\sigma,s}|\ \leq\ K\mu(S).
\end{equation}
\end{corollary}
\begin{proof}
Let $K>0$ be the constant given by Lemma 2. Fix $\sigma\in\Sigma$, a subset $S\subset[0,T]$ comprised of a
finite number of intervals, and
$s\geq\sup\{\tilde{s}\in S\}$.
We can assume without loss of generality that each one of the intervals comprising $S$ contains its lower-boundary point but not its upper-boundary point.
Denote these intervals by $I_{j}:=[s_{1,j},s_{2,j})$, $j=1,\ldots,m$ for some $m\geq 1$, so that
$S=\cup_{j=1}^{m}[s_{1,j},s_{2,j})$. Furthermore,
by subdividing these
intervals if necessary, we can assume that $v(\tau)$ has a constant value throughout each interval $I_{j}$, namely all the modes in $I_{j}$ are the same according
to $\sigma$.
Note that these intervals need not be contiguous, i.e., it is possible to have $s_{1,j+1}>s_{2,j}$ for some $j=1,\ldots,m-1$.

Define $S_{j}:=\cup_{i=1}^{j}I_{i}$, $j=1,\ldots,m$, and
note that $S=S_{m}$. Furthermore,
$\mu(S)=\sum_{j=1}^{m}(s_{2,j}-s_{1,j})$. Next, we have that
\begin{eqnarray}
D_{\sigma(S),s}-D_{\sigma,s}\ =\nonumber \\
D_{\sigma(S_{1}),s}-D_{\sigma,s}+\sum_{j=2}^{m}(D_{\sigma(S_{j}),s}-D_{\sigma(S_{j-1}),s}).
\end{eqnarray}
By Lemma 2,
$|D_{\sigma(S_{1}),s}-D_{\sigma,s}|\leq K(s_{2,1}-s_{1,1})$,
and for every
$j=2,\ldots,m$,
$|D_{\sigma(S_{j}),s}-D_{\sigma(S_{j-1}),s}|\leq K(s_{2,j}-,s_{1,j})$.
By (16)
$|D_{\sigma(S),s}-D_{\sigma,s}|\leq K\sum_{j=1}^{m}(s_{2,j}-s_{1,j})$,
and since $\mu(S)=\sum_{j=1}^{m}(s_{2,j}-s_{1,j})$, (17) follows.
\end{proof}

We now can state the algorithm's property of sufficient descent.
\begin{proposition}
Fix $\eta\in(0,1)$, $\beta\in(0,1)$, and $\alpha\in(0,\eta)$.
There exists a constant $c>0$
 such that, for every $\sigma\in\Sigma$ satisfying
$D_{\sigma}<0$, and for every
$\lambda\in[0,\mu(S_{\sigma,\eta})]$ such that $\lambda\leq c|D_{\sigma}|$,
\begin{equation}
J(\sigma(\lambda))-J(\sigma)\ \leq\ \alpha\lambda D_{\sigma}.
\end{equation}
\end{proposition}
\begin{proof}
Consider $\sigma\in\Sigma$ and an interval $I:=[s_{1},s_{2})$ such that $\sigma$ has the same mode throughout $I$. By
Lemma 1, $J(\sigma_{s_{1}}(\gamma))$ is $C^2$ in $\gamma\in[0,s_{2}-s_{1})$, and by  (4),
$J(\sigma_{s_{1}}(0))^{'}=D_{\sigma,s_{1}}$.
 Fix $a\in(\frac{\alpha}{\eta},1)$.
Suppose that $D_{\sigma,s_{1}}<0$. By Proposition 1 (Equation (9))  there exists $\xi>0$ such that, for every $\gamma\geq 0$ satisfying
$\gamma\leq\min\{-\xi D_{\sigma,s_{1}},s_{2}-s_{1}\}$,
\begin{equation}
J(\sigma_{s_{1}}(\gamma))-J(\sigma)\ \leq\ a\gamma D_{\sigma,s_{1}}.
\end{equation}
Furthermore, $\xi$ does not depend on the mode-schedule $\sigma$ or on the interval $I$.

Next, by Corollary 2 there exists a constant $K>0$ such that, for every $\sigma\in\Sigma$,  for every set
$S\subset[0,T]$ consisting of the finite union of intervals, and for every point $s\geq\sup\{\tilde{s}\in S\}$,
\begin{equation}
|D_{\sigma(S),s}-D_{\sigma,s}|\ \leq\ K\mu(S).
\end{equation}

Fix $c>0$ such that
\begin{equation}
c\ <\ \min\{\frac{2}{aK}(a\eta-\alpha),\frac{\eta}{K}\};
\end{equation}
we  next prove the assertion of the proposition for this $c$.
Fix  $\sigma\in\Sigma$ such that $D_{\sigma}<0$, and consider a set $S\subset S_{\sigma,\eta}$ consisting of the finite union of disjoint
intervals. By subdividing these intervals if necessary we can ensure that the length of each one of them is less than
$-\xi\eta D_{\sigma}$. Denote these intervals by
$I_{j}:=[s_{1,j},s_{2,j})$, $j=1,\ldots,m$ (for some $m>0$), define $\gamma_{j}:=s_{2,j}-s_{1,j}$, and define
$\lambda:=\sum_{j=1}^{M}\gamma_{j}$. Since $s_{1,j}\in S_{\sigma,\eta}$ we have that $D_{\sigma,s_{1,j}}\leq -\eta D_{\sigma}$, and
we recall that $\gamma_{j}\leq-\xi\eta D_{\sigma}$ $\forall j=1,\ldots,m$.

Next, we define the mode-schedules $\sigma^j$, $j=1,\ldots,m$, in the following recursive manner.
For $j=1$, $\sigma^1=\sigma_{s_{1,1}}(\gamma_{1})$; and for every $j=2,\ldots,m$,
$\sigma^j:=\sigma_{s_{1,j}}^{j-1}(\gamma_{j})$. In words, $\sigma^1$ is obtained from $\sigma$ by
flipping the mode at every time $s\in I_{1}$; and for every
$j=2,\ldots,m$, $\sigma^{j}$ is obtained from
$\sigma^{j-1}$ by flipping the mode at every time $s\in I_{j}$.
 Observe that $\sigma^{j}$
is also obtained from $\sigma$ by flipping the mode
at every time $s\in\cup_{i=1}^{j}I_{i}$. In
particular, $\sigma^{m}$ is obtained from $\sigma$ by flipping the modes at every time $s\in S$.
Since by assumption $\mu(S)=\sum_{j=1}^{m}\gamma_{j}=\lambda$,
we will use the notation
$\sigma(S):=\sigma(\lambda)$.

Suppose that $\lambda\leq -cD_{\sigma}$; we next establish Equation (19), and this will complete the proof.
Consider the difference-term
$J(\sigma^{j})-J(\sigma)$ for $j=1,\ldots,m$.
For $j=1$,
$J(\sigma^{1})-J(\sigma)\leq a\gamma_{1}D_{\sigma,s_{1,j}}$ (by (20)); and since $s_{1,j}\in S_{\sigma,\eta}$, $D_{\sigma,s_{1,j}}\leq\eta D_{\sigma}$, and hence
\begin{equation}
J(\sigma^{1})-J(\sigma)\ \leq\ a\gamma_{1}\eta D_{\sigma}.
\end{equation}
Nest, consider $j=2,\ldots,m$. An inequality like (23) does not necessarily hold since $\sigma^{j}$ is obtained from $\sigma$ by flipping the mode at every $s\in\cup_{i=1}^{j}I_{j}$
and $\mu(\cup_{i=1}^{j}I_{j})$ may be larger than $-\xi D_{\sigma,s_{1,j}}$, and therefore an inequality like (20) cannot be applied.
A different argument is needed.

Consider the term $J(\sigma^{j})-J(\sigma)$. Subtracting and adding $J(\sigma^{j-1})$ we obtain,
\begin{equation}
J(\sigma^{j})-J(\sigma)\ =\ J(\sigma^{j})-J(\sigma^{j-1})+J(\sigma^{j-1})-J(\sigma).
\end{equation}
Now $\sigma^{j}$ is obtained from $\sigma^{j-1}$ by flipping the mode at every time $s\in I_{j}$ and hence $\sigma^{j}=\sigma_{s_{1,j}}^{j-1}(\gamma_{j})$,
while $\sigma^{j-1}=\sigma_{s_{1,j}}^{j-1}(0)$ since
in the latter term no mode is being flipped. Therefore,
\begin{equation}
J(\sigma^{j})-J(\sigma^{j-1})\ =\ J(\sigma_{s_{1,j}}^{j-1}(\gamma_{j}))-J(\sigma_{s_{1,j}}^{j-1}(0)).
\end{equation}
We next show that $D_{\sigma^{j-1},s_{1,j}}<0$ in order to be able to use Equation (20).
By (17),
\begin{equation}
|D_{\sigma^{j-1},s_{1,j}}-D_{\sigma,s_{1,j}}|\ \leq\ K\Sigma_{i=1}^{j-1}\gamma_{i}.
\end{equation}
By definition $\Sigma_{i=1}^{j-1}\gamma_{i}\leq\sum_{i=1}^{m}\gamma_{i}=\lambda$; by assumption $\lambda\leq c|D_{\sigma}|$; and by (22) $K\leq\frac{\eta}{c}$;
consequently, and by (26),
$|D_{\sigma^{j-1},s_{1,j}}-D_{\sigma,s_{1,j}}|\leq \eta |D_{\sigma}|$.
But $s_{1,j}\in S_{\sigma,\eta}$ and hence $D_{\sigma,s_{1,j}}\leq\eta D_{\sigma}$, and this
implies that $D_{\sigma^{j-1},s_{1,j}}\leq 0$.

An application of (20) to (25) now yields that
\begin{equation}
J(\sigma^j)-J(\sigma^{j-1})\ \leq\ a\gamma_{j}D_{\sigma^{j-1},s_{1,j}}.
\end{equation}
We do not know whether or not $D_{\sigma^{j-1},s_{1,j}}\leq \eta D_{\sigma}$, but we know that
$D_{\sigma,s_{1,j}}\leq \eta D_{\sigma}$ (since
$s_{1,j}\in S_{\sigma,\eta})$. Applying
(26)  to (27) we  obtain that
\begin{eqnarray}
J(\sigma^{j})-J(\sigma^{j-1})\ \leq\ a\gamma_{j}D_{\sigma^{j-1},s_{1,j}}\ =\nonumber \\
 a\gamma_{j}D_{\sigma,s_{1,j}}+a\gamma_{j}(D_{\sigma^{j-1},s_{1,j}}-D_{\sigma,s_{1,j}})\ \leq\nonumber \\  a\gamma_{j}D_{\sigma,s_{1,j}}+a\gamma_{j}K\sum_{i=1}^{j-1}\gamma_{i}.
\end{eqnarray}
But $D_{\sigma,s_{1,j}}\leq\eta D_{\sigma}$ (since $s_{1,j}\in S_{\sigma,\eta}$), and hence,
$J(\sigma^{j})-J(\sigma^{j-1})\leq a\gamma_{j}\eta D_{\sigma}+aK\gamma_{j}\sum_{i=1}^{j-1}\gamma_{i}$.
Using this inequality in (24) yields the following one,
\begin{equation}
J(\sigma^{j})-J(\sigma)\ \leq\ a\gamma_{j}\eta D_{\sigma}+aK\gamma_{j}\sum_{i=1}^{j-1}\gamma_{i}+J(\sigma^{j-1})-J(\sigma).
\end{equation}
Apply (29) repeatedly and recursively with $j=1,\ldots,m$ to obtain, after some algebra, the following inequality:
\begin{equation}
J(\sigma^{m})-J(\sigma)\ \leq\ a(\sum_{i=1}^{m}\gamma_{i})\eta D_{\sigma}+aK\sum_{i,\ell=1,i\neq\ell}^{m}\gamma_{i}\gamma_{\ell}.
\end{equation}
But $\sum_{i=1}^{m}\gamma_{i}=\lambda$,
$\sum_{i,\ell=1,i\neq\ell}^{m}\gamma_{i}\gamma_{\ell}\leq\frac{1}{2}(\sum_{i=1}^{m}\gamma_{i})^2=\frac{1}{2}\lambda^2$, and
$\sigma^{m}:=\sigma(\lambda)$, and hence,
\begin{equation}
J(\sigma(\lambda))-J(\sigma)\ \leq\ a\lambda\eta D_{\sigma}+\frac{1}{2}aK\lambda^2.
\end{equation}
By assumption $\lambda\leq -cD_{\sigma}$, and by (22)
$aKc<2(a\eta-\alpha)$, and this, together with (31), implies (19). The proof is now complete.
\end{proof}

General results concerning sufficient descent, analogous to
Proposition 2, provide key arguments in proving asymptotic convergence of nonlinear-programming  algorithms (see, e.g., \cite{Polak97}). In our case, the optimality function has the peculiar property that it is discontinuous in the Lebesgue measure of the set where a mode is flipped. To see this,
recall that $D_{\sigma,s,v(s)}=p(s)^T\big(f(x(s),v(s)^c)-f(x(s),v(s))\big)$ (see Equation (4)),
and hence a change of the mode at time $s$ would flip the sign of $D_{\sigma,s,v(s)}$. This can result in situations where  $|D_{\sigma}|$
is ``large'' while   $S_{\sigma,\eta}$ is ``small'', and for this reason, convergence of Algorithm 1 is characterize by Equation (5) with the  $limsup$ rather than
with the stronger assertion with $lim$. This is the subject of the following result.
\begin{corollary}
Suppose that Algorithm 1 computes a sequence of schedules,
$\{\sigma_{k}\}_{k=1}^{\infty}$. Then Equation (5) is in force, namely
$\limsup_{k\rightarrow\infty}D_{\sigma_{k}}=0$.
\end{corollary}
\begin{proof}
Suppose, for the sake of contradiction, that Equation (5) does not hold. Then the sufficient-descent property proved in Proposition 2  implies that $\lim_{k\rightarrow\infty}\lambda(\sigma_{k})=0$, for
otherwise
Equation (19) would yield $\lim_{k\rightarrow\infty}J(\sigma_{k})=-\infty$ which is impossible.
Next, By the definition of $\lambda(\sigma_{k})$ (Step 3 of the algorithm),
there exists $k_{0}$ such that $\forall k\geq k_{0}$, $\lambda(\sigma_{k})=\mu(S_{\sigma_{k},\eta})$, and hence $\lim_{k\rightarrow\infty}\mu(S_{\sigma_{k},\eta})=0$.
In this case, $\sigma_{k+1}$ is obtained from $\sigma_{k}$ by flipping the modes at every $s\in S_{\sigma_{k},\eta}$. By the perturbation theory of differential equations
(e.g., Proposition 5.6.7 in \cite{Polak97}), $x$ and $p$ are Lipschitz continuous in their $L^{\infty}$ norms with respect to
the Lebesgue measure of the sets where the modes are flipped, i.e. $\mu(S_{\sigma_{k},\eta})$. Therefore,
and by (4) and the definition of $S_{\sigma_{k},\eta}$, there exist $k_{1}\geq k_{0}$ and $\zeta\in(0,1)$ such that
$\forall k\geq k_{1}$, $D_{\sigma_{k+1},\eta}\geq \zeta D_{\sigma_{k},\eta}$, implying   that
$\lim_{k\rightarrow\infty}D_{\sigma_{k},\eta}=0$. However this is a contradiction to the assumption that (5) does not hold, thus completing the proof.
\end{proof}

Alternative optimality functions can be considered as well, like the term
$D_{\sigma}\mu(S_{\sigma,\eta})$, where it is apparent  Equation (19)   that
$\lim_{k\rightarrow\infty}D_{\sigma_{k}}\mu(S_{\sigma_{k},\eta})=0$. The choice of the ``most appropriate''
optimality function is an interesting theoretical question that will be addressed elsewhere, while here we consider
the simplest and (in our opinion)
most intuitive optimality function $D_{\sigma}$,
despite its technical peculiarities.

Finally, a word must be said about the general case where the set $V$ consists of more than two points.
The algorithm and much of its analysis remain unchanged, except that for a given $\sigma\in\Sigma$,
at a time $s$, the mode associated with $v(s)$ should be switched to the
mode associated with the point $w\in V$ that minimizes the term
$D_{\sigma,s,w}$.

\section {Numerical Example}

We tested the algorithm on the double-tank system shown in Figure 1. The input to the system,
$v$,
is the inflow rate to the upper tank, controlled by the valve and having two possible values,
$v_{1}=1$ and $v_{2}=2$. $x_{1}$ and $x_{2}$ are the fluid levels at the upper tank and lower tank, respectively, as shown in the figure. According to Toricelli's law, the state equation
is
\begin{equation}
\left(
\begin{array}{c}
\dot{x}_{1}\\
\dot{x}_{2}
\end{array}
\right)
\  =\ \left(
\begin{array}{c}
v-\sqrt{x_{1}}\\
\sqrt{x_{1}}-\sqrt{x_{2}}
\end{array}
\right),
\end{equation}
with the (chosen)  initial condition  $x_{1}(0)=x_{2}(0)=2.0$.
Notice that both $x_{1}$ and $x_{2}$ must satisfy the inequalities $1\leq x_{i}\leq 4$,
and if $v=1$ indefinitely than $\lim_{t\rightarrow\infty}x_{i}=1$, while
if $v=2$ indefinitely then $\lim_{t\rightarrow\infty}x_{i}(t)=4$, $i=1,2$.

\begin{figure}
\centering
\epsfig{file=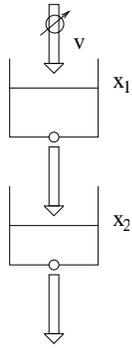,width=1.6cm}\\
 \caption{Two-tank system}
\end{figure}

The objective of the optimization problem is to have the fluid level in the lower tank
track the given value of 3.0, and hence we chose the performance criterion to be
\begin{equation}
J\ =\ 2\int_{0}^{T}\big(x_{2}-3\big)^2dt,
\end{equation}
for the final-time $T=20$. The various integrations were computed by the forward-Euler method with
$\Delta t=0.01$. For the algorithm we chose the parameter-values $\alpha=\beta=0.5$ and
$\eta=0.6$, and we ran it from the initial mode-schedule associated with the control input
$v(t)=1\ \forall\ t\in[0,10]$ and  $v(t)=2\ \forall\ t\in(10,20]$.

Results of a typical run, consisting of 100 iterations of the
 algorithm, are shown in Figures 2-5. Figure 2 shows the control computed after 100 iterations, namely the input control $v$ associated with $\sigma_{100}$.
 The graph is not surprising, since we expect the optimal control initially to consist of $v=2$
 so that $x_{2}$ can rise to a value close to 3, and then to enter a sliding mode
 in order for $x_{2}$ to maintain its proximity to 3. This is evident from Figure 2,
 where the sliding mode has begun to be constructed. Figure 3 shows the resulting
 state trajectories $x_{1}(t)$ and $x_{2}(t)$, $t\in[0,T]$, associated with the last-computed schedule  $\sigma_{100}$.
 The jagged curve is of $x_{1}$ while the smoother curve is of $x_{2}$.
 It is evident that $x_{2}$ climbs towards 3 initially and tends to stay there thereafter.
 Figure 4 shows the graph of the cost criterion $J(\sigma_{k})$ as a function of the iteration count $k=1,\ldots,100$.
 The initial schedule, $\sigma_{1}$, is far away from the minimum and its associated cost is $J(\sigma_{1})=70.90$, and the cost of the last-computed schedule
 is $J(\sigma_{100})=4.87$. Note that
 $J(\sigma_{k})$ goes down to under 8 after 3 iterations. Figure 5 shows the optimality function
 $D_{\sigma_{k}}$ as a function of the iteration count $k$. Initially $D_{\sigma_{1}}=-14.92$
 while at the last-computed schedule $D_{\sigma_{100}}=-0.23$, and it is seen that $D_{\sigma_{k}}$ makes significant
 climbs towards 0 in  few iterations. We also ran the algorithm for 200 iterations
 from the same initial schedule $\sigma_{1}$, in order to verify that $J(\sigma_{k})$ and $D_{\sigma_{k}}$
 stabilize. Indeed they do, and $J$ declined from $J(\sigma_{100})=4.87$ to $J(\sigma_{200})=4.78$, while the optimality functions continues to rise towards 0, from $D_{\sigma_{100}}=-0.23$ to $D_{\sigma_{200}}=-0.062$.

\begin{figure}
\centering
\epsfig{file=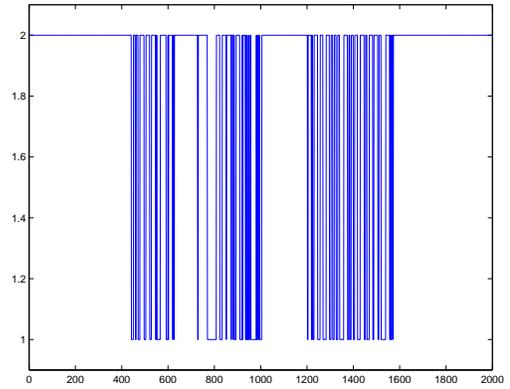,width=6.6cm}\\
 \caption{Control (schedule) obtained after 100 iterations}
\end{figure}

\begin{figure}
\centering
\epsfig{file=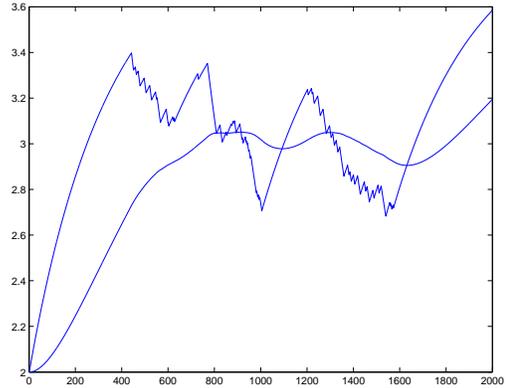,width=6.6cm}\\
 \caption{$x_{1}$ and $x_{2}$ vs. $t$}
\end{figure}

\begin{figure}
\centering
\epsfig{file=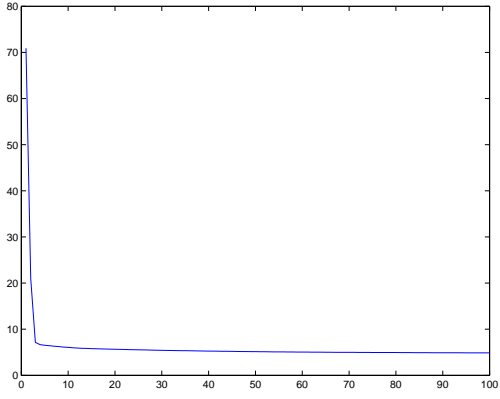,width=6.6cm}\\
 \caption{Cost criterion vs. iteration count}
\end{figure}

\begin{figure}
\centering
\epsfig{file=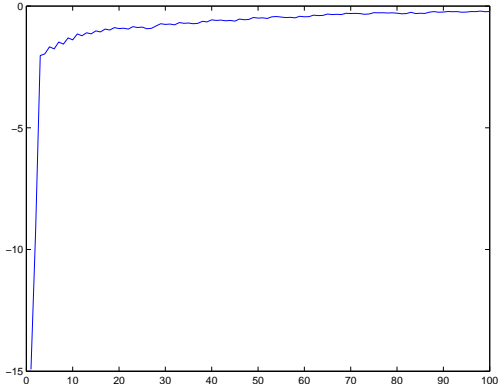,width=6.6cm}\\
 \caption{Optimality function vs. iteration count}
\end{figure}

\section{Conclusions}
This paper proposes a new algorithm for the optimal mode-scheduling problem, where it is desirable to minimize
an integral-cost criterion defined on the system's state trajectory as a function of the modes' schedule.
The algorithm is based on the principle of gradient descent with Armijo step sizes, comprised of the Lebesgue measures of sets where the modes are
being changed. Asymptotic convergence is proved in the sense of minimizing sequences, and simulation results support the theoretical developments.
Future research will refine the proposed algorithmic framework and apply it to large-scale problems.

\section{Appendix}

The purpose of this appendix is to provide proofs to Proposition 1, and Lemmas 1 and 2.

{\it Proof of Proposition 1.}

(1). The main argument is based on the following form of the second-order Taylor series expansion:
For every $x\in R^n$ and $y\in R^n$,
\begin{eqnarray}
f(x+y)-f(x)\ =\ \langle\nabla f(x),y\rangle\ +\nonumber \\
\int_{0}^{1}(1-\xi)\langle H(x+\xi y)y,y\rangle d\xi,
\end{eqnarray}
where $\langle\cdot\rangle$ denotes inner product in $R^n$. Apply this with $y=-\lambda h(x)$
to obtain,
\begin{eqnarray}
f(x-\lambda h(x))-f(x)\ =\ -\lambda\langle\nabla f(x),h(x)\rangle\ +\nonumber \\
\lambda^2\int_{0}^{1}(1-\xi)\langle H(x-\xi\lambda h(x))h(x),h(x)\rangle d\xi.
\end{eqnarray}
Add $\alpha\lambda||\nabla f(x)||$ to both sides of this equation, and use the fact that $||H(\cdot)||\leq L$, to obtain (after some algebra) that
\begin{eqnarray}
f(x-\lambda h(x))-f(x)+\alpha\lambda||\nabla f(x)||\nonumber \\
\leq\ -\lambda\big((1-\alpha)||\nabla f(x)||-\frac{\lambda}{2}L\big).
\end{eqnarray}
Now if $0\leq \lambda\leq \frac{2}{L}(1-\alpha)||\nabla f(x)||$ then the Right-Hand side of (36) is non-positive,  hence Equation (9) is satisfied.

(2). Follows directly from Part (1), Equation (7), and the definition of $\lambda(x)$ (8).\qed

The proofs of Lemma 1 and Lemma 2 follow as corollaries from established results on sensitivity analysis of solutions to differential equations,
presented in  Section 5.6 of
\cite{Polak97}. In fact, the results of interest here involve mode-insertions via needle variations, which is a special case of the
setting in \cite{Polak97} where general variations in the control are considered. Furthermore, the perturbations here are parameterized by
a one-dimensional variable and hence the results are in terms of derivatives in the usual sense, while those in \cite{Polak97} are
 in terms of G\^ateaux or Fr\'echet derivatives.

{\it Proof of Lemma 1.} By Proposition 5.6.5 in \cite{Polak97} and the Bellman-Gronwall Lemma,
the terms $||x(t)||_{L^{\infty}}$ are uniformly bounded over the space of controls $v$ associated with every $\sigma\in\Sigma$.
The costate equation (3) yields a similar result for
$||p(t)||_{L^{\infty}}$. Next, recall that $v(\cdot)$ has a constant value throughout the interval
$[s_{1},s_{2}]$, and hence the differentiability assumptions of Theorem 5.6.10 in \cite{Polak97} are valid.
This theorem implies that
$J(\sigma_{s_{1}}(\gamma))^{''}$ exists and is expressed in terms of the Hamiltonian and its first two derivatives, hence it is uniformly bounded.\qed

{\it Proof of Lemma 2.} Since $v(\cdot)$ has a constant value throughout the interval $[s_{1},s_{2}]$, the assumptions made in the statement of Lemma 5.6.7
in \cite{Polak97} are in force. This implies a uniform Lipschitz continuity of $x$ and $p$ with respect to variations in $\gamma$. In the setting of Lemma 2, the needle
variations is made at the same point $s\geq s_{2}$ for both mode-schedules $\sigma$ and $\sigma_{s_{1}}(\gamma)$, and hence, and by Equation (4),
Equation (16) follows.\qed

\end{document}